\documentclass[abstract,headings=normal]{scrartcl}

\usepackage[automark]{scrpage2}
\usepackage[english]{babel}
\usepackage[T1]{fontenc}
\usepackage[applemac]{inputenc}
\usepackage{graphicx}
\usepackage{amsmath}
\usepackage{amssymb}
\usepackage{amsthm}
\usepackage{subfig}
\usepackage{hyperref}

\newcommand{\C}{\mathbb{C}}
\newcommand{\Z}{\mathbb{Z}}

\newcommand{\Ell}{\mathcal{L}}

\newcommand{\Hvier}{H\textsuperscript{4}}
\newcommand{\Hsechs}{H\textsuperscript{6}}

\newtheorem{theo}{Theorem}[section]
\newtheorem{lemma}[theo]{Lemma}
\theoremstyle{definition}

\theoremstyle{remark}
\newtheorem*{rem}{Remark}

\graphicspath{{./graphics/}}
\DeclareGraphicsExtensions{.pdf}

\begin{document}

\title{On the Lagrangian structure of \\ 3D consistent systems of asymmetric quad-equations}

\author{Raphael Boll\and Yuri B. Suris
\footnote{Institut f\"ur Mathematik, MA 7-2, Technische Universit\"at Berlin, Str.~des~17.~Juni~136, 10623 Berlin, Germany; e-mail: boll@math.tu-berlin.de, suris@math.tu-berlin.de}}

\maketitle

\begin{abstract}
Recently, the first-named author gave a classification of 3D consistent 6-tuples of quad-equations with the tetrahedron property; several novel asymmetric 6-tuples have been found. Due to 3D consistency, these 6-tuples can be extended to discrete integrable systems on $\mathbb Z^m$. We establish Lagrangian structures and flip-invariance of the action functional for the class of discrete integrable systems involving equations for which some of the biquadratics are non-degenerate and some are degenerate. This class covers, among others, some of the above mentioned novel systems.
\end{abstract}

\section{Introduction}
One of the definitions of integrability of discrete equations, which becomes increasingly popular in the recent years, is based on the notion of multidimensional consistency. For two-dimensional systems, this notion was clearly formulated first in \cite{NW}, and it was pushed forward as a synonym of integrability in \cite{quadgraphs,Nijhoff}. The outstanding importance of 3D consistency in the theory of discrete integrable systems became evident no later than with the appearance of the well-known ABS-classification of integrable equations in \cite{ABS1}. This classification deals with quad-equations which are set on the faces of an elementary cube in such a manner that all faces carry similar equations differing only by the parameter values assigned to the edges of a cube. Moreover, each single equation admits a $D_4$ symmetry group. Thus, ABS-equations from \cite{ABS1} can be extended in a straightforward manner to the whole of $\Z^{m}$.  In \cite{ABS2}, a relaxed definition of 3D consistency was introduced: the faces of an elementary cube are allowed to carry a priori different quad-equations. The classification performed in \cite{ABS2} is restricted to so-called equations of type~Q, i.e., those equations whose biquadratics are all non-degenerate (a precise definition will be recalled in Section \ref{faces}). In \cite{Atk1} and \cite{Todapaper}, numerous asymmetric systems of quad-equations have been studied, which also include equations of type H (with degenerate biquadratics). A classification of such systems has been given in  \cite{classification}.  This classification covers the majority of systems from \cite{Atk1} and all systems from \cite{Todapaper}, as well as equations from \cite{Hietarinta}, \cite{LY} and \cite{HV}. Moreover, it contains also a number of novel systems. The results of this local classification lead to integrable lattice systems via a procedure of reflecting the cubes in a suitable way (see \cite{classification} for details).
As shown in \cite{Atk1}, \cite{Todapaper},  \cite{classification}, the asymmetric systems still can be seen as families of B\"acklund transformations, and they lead to zero curvature representations of participating quad-equations.

After the pioneering work \cite{MV} the variational (Lagrangian) structure of discrete integrable systems is a topic which receives more and more attention. A Lagrangian formulation of systems from the ABS list is given in \cite{ABS1}. The new idea in \cite{LN} was to extend the action functional of \cite{ABS1} to a multidimensional lattice, which is meaningful due to the invariance of the action under elementary 3D flips of 2D quad-surfaces in $\Z^{m}$. This remarkable and fundamental flip invariance property was proven in \cite{LN} for some particular cases from the ABS list and in \cite{BS1} for the whole list.

In the present paper, we introduce a Lagrangian formulation for an important subclass of systems from the classification in \cite{classification}, namely those which include type H equations not all of whose biquadratics are degenerate, and prove the flip invariance of action functionals for these systems.

The outline of the paper is as follows: In Sections~\ref{faces} and \ref{cubes} we introduce the necessary notions and notation and recapitulate the classification results from \cite{classification} which are relevant for the present paper. Then, in Section~\ref{threelegforms} we use the so-called three-leg form of quad-equations to construct the Lagrangians for our equations. The fundamental property of Lagrangians on an elementary quadrilateral constitutes one of the main results of the paper (Theorem \ref{th:single}). The Lagrangians are elementary building blocks of the action functional constructed in Section~\ref{flipQuad}, which also contains the second main result of the paper, namely  the flip invariance of the action functional for systems of our class.

\section{Quad-equations on a single square} \label{faces}

We start with introducing relevant objects and notations. A \emph{quad-equation} is a relation of the type $Q\left(x_{1},x_{2},x_{3},x_{4}\right)=0$, where $Q\in\C\left[x_{1},x_{2},x_{3},x_{4}\right]$ is an irreducible multi-affine polynomial. It is convenient to visualize a quad-equation by an elementary square whose vertices carry the four fields $x_1,x_2,x_3,x_4$, cf. Figure~\ref{fig:2}. For quad-equations without pre-supposed symmetries, the natural classification problem is posed modulo the action of four independent M\"obius transformations of all four variables:
\[
Q(x_1,x_2,x_3,x_4)\; \rightsquigarrow\;
\prod_{k=1}^4(c_kx_k+d_k)\cdot Q\Big(\frac{a_1x_1+b_1}{c_1x_1+d_1},\frac{a_2x_2+b_2}{c_2x_2+d_2},
\frac{a_3x_3+b_3}{c_3x_3+d_3},\frac{a_4x_4+b_4}{c_4x_4+d_4}\Big).
\]
In the simplest situation, a quad-equation is thought of as an elementary building block of a discrete system $Q(x_{m,n},x_{m+1,n},x_{m+1,n+1},x_{m,n+1})=0$ for a function $x:\Z^2\to\C$, but we will also encounter more tricky ways of composing discrete systems on $\Z^2$ from different quad-equations.

A very useful tool for studying and classifying quad-equations, or, better, multi-affine polynomials of four variables, are their \emph{biquadratics}. For any two (different) indices $i,j\in\{1,2,3,4\}$, we define the corresponding biquadratic as follows:
\[
Q^{i,j}=Q^{i,j}\left(x_{i},x_{j}\right)=Q_{x_{k}}Q_{x_{\ell}}-QQ_{x_{k},x_{\ell}},
\]
where $\{k,\ell\}$ is the complement of $\{i,j\}$ in $\{1,2,3,4\}$. A biquadratic polynomial $h\left(x,y\right)$ is called \emph{non-degenerate} if no polynomial in its equivalence class with respect to M\"obius transformations in $x$ and $y$ is divisible by a factor $x-\alpha_{1}$ or $y-\alpha_{2}$ (with $\alpha_{i}\in\C$). Otherwise, $h$ is called \emph{degenerate}. Here, of course, the action of two M\"obius transformations on $h$ is defined as
\[
h(x,y)\; \rightsquigarrow\;
(c_1x+d_1)^2(c_2y+d_2)^2\cdot h\Big(\frac{a_1x+b_1}{c_1x+d_1},\frac{a_2y+b_2}{c_2y+d_2}\Big).
\]
We recall also that an important role in studying the multi-affine polynomials $Q$ belongs to the four discriminants
\[
r^i(x_i)=(Q^{i,j}_{x_j})^2-2Q^{i,j}_{x_jx_j}Q^{i,j}
\]
(it can be shown that the result does not depend on $j$). These are quartic polynomials which can be assigned to the vertices of an elementary square. Note that the action of a M\"obius transformation on a quartic polynomial $r(x)$ is defined as
\[
r(x)\;\rightsquigarrow\;(cx+d)^4\cdot r\Big(\frac{ax+b}{cx+d}\Big).
\]

A multi-affine polynomial $Q$ is said to be of \emph{type~Q} if all its biquadratics are non-degenerate and of \emph{type~H} otherwise. It was proven in \cite{classification} that in the latter case either four out of six biquadratics are degenerate, and then $Q$ is said to be of \emph{type~\Hvier}, or all six biquadratics are degenerate, and then $Q$ is said to be of \emph{type~\Hsechs}.

A complete classification of multi-affine polynomials $Q$ up to M\"obius transformations in $x_{1},\ldots,x_{4}$ is given in \cite{classification}. In the present paper, we will be only interested in 3D consistent six-tuples of quad-equations containing at least one equation of type~\Hvier\ but no equations of type~\Hsechs. (Recall that six-tuples consisting solely of type~Q equations were already classified in \cite{ABS2}, and the flip invariance of the action functional for such systems was completely discussed in \cite{BS1}.) As shown in \cite{classification}, the list of canonical normal forms of polynomials of  type~\Hvier\ consists of three items (which actually appeared already in \cite{ABS2}):
\begin{itemize}
\item[$H_{1}^{\epsilon}$:] $Q=\left(x_{1}-x_{3}\right)\left(x_{2}-x_{4}\right)+
    \left(\alpha_{2}-\alpha_{1}\right)\left(1-\epsilon^{2} x_{2}x_{4}\right)$,\par
\item[$H_{2}^{\epsilon}$:] $Q=\left(x_{1}-x_{3}\right)\left(x_{2}-x_{4}\right)+
    \left(\alpha_{2}-\alpha_{1}\right)\left(x_{1}+x_{2}+x_{3}+x_{4}\right)+
    \alpha_{2}^{2}-\alpha_{1}^{2}$\par
$\qquad-\dfrac{\epsilon}{2}\left(\alpha_{2}-\alpha_{1}\right)\left(2 x_{2}+\alpha_{1}+\alpha_{2}\right)\left(2 x_{4}+\alpha_{1}+\alpha_{2}\right)-\dfrac{\epsilon}{2}\left(\alpha_{2}-\alpha_{1}\right)^{3}$,\par
\item[$H_{3}^{\epsilon}$:]
    $Q=e^{2\alpha_{1}}\left(x_{1}x_{2}+x_{3}x_{4}\right)-
    e^{2\alpha_{2}}\left(x_{1}x_{4}+x_{2}x_{3}\right)+
    \left(e^{4\alpha_{1}}-e^{4\alpha_{2}}\right)\left(\delta^{2}+
    \dfrac{\epsilon^{2} x_{2}x_{4}}{e^{2\left(\alpha_{1}+\alpha_{2}\right)}}\right)$.
\end{itemize}
Thus, the members of this list, like those of the list Q, depend on two parameters $\alpha_1,\alpha_2$ which can be assigned to the pairs of opposite edges of the elementary square. However, unlike the Q case, where all vertices are on the equal footing (which can be visualized as on Figure \ref{fig:3}), for any type \Hvier\ polynomial, the four discriminants assigned to the vertices belong to two different types (at least for $\epsilon\neq 0$). Up to constant factors $\kappa^2(\alpha_1,\alpha_2)$ depending on the parameters only and listed in Section \ref{TypeH} (cf. also Lemma \ref{lem:1param}), one has:
\begin{itemize}
\item[$H_{1}^{\epsilon}$:]
    \quad $r^1(x_1)=\epsilon^2$,\; $r^2(x_2)=0$,\;  $r^3(x_3)=\epsilon^2$,\; $r^4(x_4)=0$   ; \par
\item[$H_{2}^{\epsilon}$:]
   \quad $r^1(x_1)=1+4\epsilon x_1$,\; $r^2(x_2)=1$,\; $r^3(x_3)=1+4\epsilon x_3$,\; $r^4(x_4)=1$; \par
\item[$H_{3}^{\epsilon}$:]
   \quad  $r^1(x_1)=x_1^2-4\delta^2\epsilon^2$,\; $r^2(x_2)=x_2^2$,\; $r^3(x_3)=x_3^2-4\delta^2\epsilon^2$,\; $r^4(x_4)=x_2^2$.
\end{itemize}
This leads to the coloring of the vertices in two black and two white ones. We will think of the black vertices as carrying the parameter $\epsilon$, while for the white vertices $\epsilon=0$.


It will be useful to visually keep track of the distribution of the four degenerate and two non-degenerate biquadratics of a type \Hvier\ polynomial. We will indicate the two non-degenerate ones by thick lines. When composing discrete systems from quad-equations, it becomes crucial how to arrange the variables (fields) $x_{1}$, $x_{2}$, $x_{3}$ and $x_{4}$ over the vertices of an elementary square. There are two different (up to rotations of an elementary square) quad-equations of type~\Hvier, corresponding to two different \emph{biquadratics patterns}, as shown on Figure~\ref{fig:2}.

\begin{figure}[htbp]
   \centering
   \subfloat[Quad-equation of type Q]
    {\label{fig:3}\includegraphics[scale=1.2]{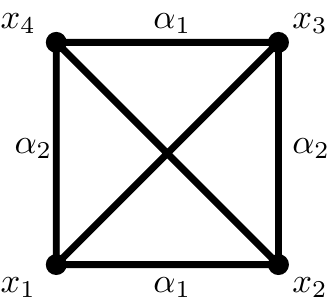}}\qquad
   \subfloat[Rhombic version of a quad-equation of type \Hvier]
     {\label{fig:2a}\includegraphics[scale=1.2]{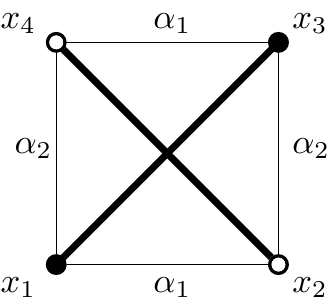}}\qquad
   \subfloat[Trapezoidal version of a quad-equation of type \Hvier]
     {\label{fig:2b}\includegraphics[scale=1.2]{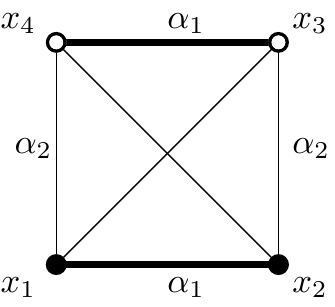}}
   \caption{Biquadratics patterns of quad-equations; non-degenerate biquadratics are indicated by thick lines}
\label{fig:2}
\end{figure}

The first arrangement corresponds to assigning the two black vertices to one diagonal of the elementary square and the two white vertices to another diagonal, see Figure \ref{fig:2a}. We will call these quad-equations the {\em rhombic version of $H_k^\epsilon$}, due to the symmetry properties
\[
Q\left(x_{1},x_{2},x_{3},x_{4};\alpha_{1},\alpha_{2},\epsilon\right)=
-Q\left(x_{3},x_{2},x_{1},x_{4};\alpha_{2},\alpha_{1},\epsilon\right)=
-Q\left(x_{1},x_{4},x_{3},x_{2};\alpha_{2},\alpha_{1},\epsilon\right)
\]
of the corresponding polynomials.

The second version of a quad-equation is obtained from the canonical normal form $H_k^\epsilon$ by
the switching $x_2\leftrightarrow x_3$ and $\alpha_1\leftrightarrow \alpha_2-\alpha_1$. As a consequence, the two black vertices share an edge rather than a diagonal, and the same holds true for the two white vertices. This version of a quad-equation possesses only one axis symmetry which interchanges the vertices of the same color (see Figure \ref{fig:2b}). We will call the corresponding quad-equations the {\em trapezoidal version of $H_k^\epsilon$}.

By the way, the biquadratics pattern provides us with an unambiguous way of bi-coloring the vertices of the square even in the case $\epsilon=0$, when all four discriminants are the same (this corresponds to the list H from \cite{ABS1}): each of the two thick lines corresponding to the two non-degenerate biquadratics connects two vertices of the same color.

Concerning discrete systems composed of quad-equations, we mention that the rhombic version of any equation $H_k^\epsilon$ can be imposed on any bipartite quad-graph (cf. also \cite{XP}). As for the trapezoidal version of $H_k^\epsilon$, it can be imposed on the regular square lattice $\Z^2$ with horizontal (or vertical) rows colored alternately black and white, see Figure \ref{fig:6}. Integrability of such a non-autonomous system of quad-equations understood as 3D consistency has been discussed in \cite{classification}.

\begin{figure}[htbp]
   \centering
   \includegraphics[scale=1]{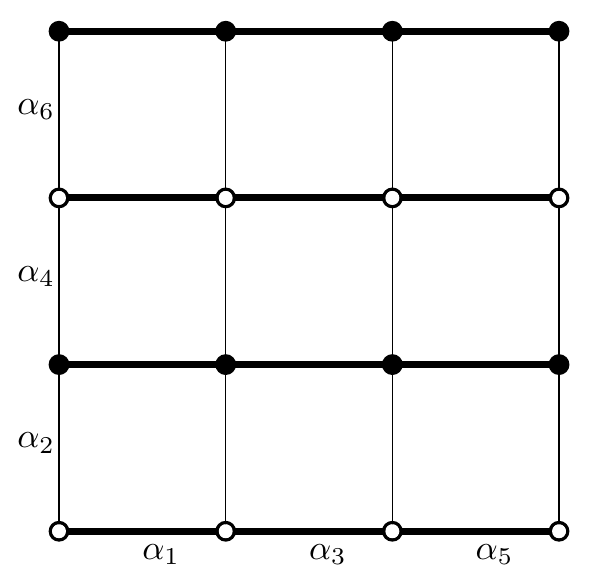}
   \caption{Integrable system on $\mathbb Z^2$ consisting of quad-equations $H_k^\epsilon$}
\label{fig:6}
\end{figure}

\section{3D consistent six-tuples of quad-equations} \label{cubes}

We will now consider six-tuples of (a priori different) quad-equations assigned to the faces of a 3D cube:
\begin{eqnarray} \label{system}
A\left(x,x_{1},x_{2},x_{12}\right)=0,& \quad &
\bar{A}\left(x_{3},x_{13},x_{23},x_{123}\right)=0,\nonumber\\
B\left(x,x_{2},x_{3},x_{23}\right)=0,& \quad &
\bar{B}\left(x_{1},x_{12},x_{13},x_{123}\right)=0,\\
C\left(x,x_{1},x_{3},x_{13}\right)=0,& \quad &
\bar{C}\left(x_{2},x_{12},x_{23},x_{123}\right)=0,\nonumber
\end{eqnarray}
see Figure~\ref{fig:cube}.
Such a six-tuple is \emph{3D consistent} if, for arbitrary initial data $x$, $x_{1}$, $x_{2}$ and $x_{3}$,
the three values for $x_{123}$ (calculated by using $\bar{A}=0$, $\bar{B}=0$ or $\bar{C}=0$) coincide. A 3D consistent six-tuple is said to possess the \emph{tetrahedron property} if there exist two polynomials $K$ and $\bar{K}$ such that the equations
\[
K\left(x,x_{12},x_{13},x_{23}\right)=0,\qquad \bar{K}\left(x_{1},x_{2},x_{3},x_{123}\right)=0
\]
are satisfied for every solution of the six-tuple. It can be shown that the polynomials $K$ and $\bar{K}$ are multi-affine and irreducible (see \cite{classification}).
\begin{figure}[htbp]
   \centering
   \subfloat[A 3D consistent six-tuple]{\label{fig:cube}\includegraphics{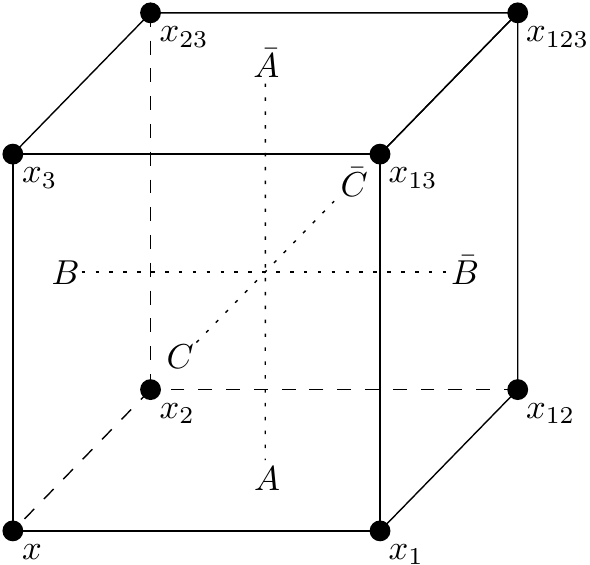}}\qquad
   \subfloat[Making tetrahedra to faces]{\label{fig:cube2}\includegraphics{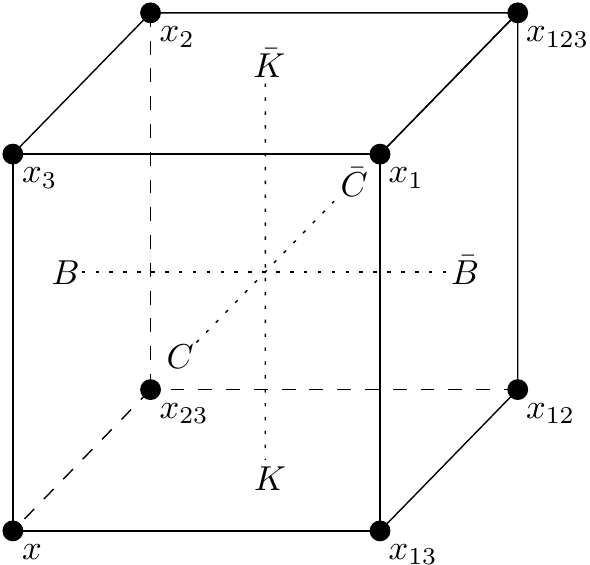}}
   \caption{Equations on a cube}
\end{figure}

A complete classification of 3D consistent six-tuples~\eqref{system} possessing the tetrahedron property, modulo (possibly different) M\"obius transformations of the eight fields $x$, $x_{i}$, $x_{ij}$, $x_{123}$, is given in \cite{classification}. We mention one of the important general properties of such six-tuples:
\begin{lemma}\label{lem:opposite}
Equations on opposite faces are of the same type and have the same biquadratics pattern.
\end{lemma}
In the present article, we consider only six-tuples~\eqref{system} whose equations are not all of type~Q and which do not contain type~\Hsechs\ equations. Taking into account Lemma \ref{lem:opposite}, one easily sees that three different combinatorial arrangements of biquadratics patterns are possible, as indicated on Figure~\ref{fig:5}.
\begin{figure}[htbp]
   \centering
   \subfloat[First case: all face equations of type \Hvier, tetrahedron equations of type Q]{\label{fig:5a}\includegraphics[scale=1.2]{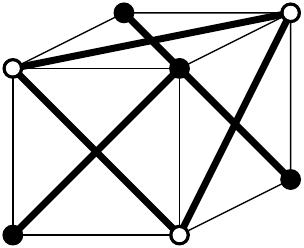}}\qquad
   \subfloat[Second case: two pairs of face equations and tetrahedron equations of type \Hvier, one pair of face equations of type Q]{\label{fig:5b}\includegraphics[scale=1.2]{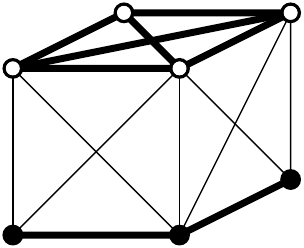}}\qquad
   \subfloat[Third case: all face and tetrahedron equations of type \Hvier]{\label{fig:5c}\includegraphics[scale=1.2]{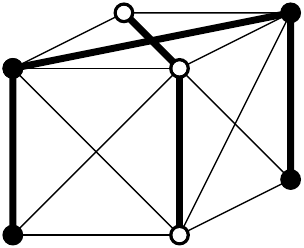}}
   \caption{Biquadratics patterns}
   \label{fig:5}
\end{figure}
Arrangements on Figures~\ref{fig:5a} and \ref{fig:5b} are related to each other by the following construction, see \cite{classification} for a proof:
\begin{lemma} \label{tetrahedronUse}
Consider a 3D consistent six-tuple~\eqref{system} possessing the tetrahedron property expressed by the two equations
\[
K\left(x,x_{12},x_{13},x_{23}\right)=0,\qquad
\bar{K}\left(x_{1},x_{2},x_{3},x_{123}\right)=0.
\]
Then the six-tuple
\begin{eqnarray}\label{dualsystem}
K\left(x,x_{12},x_{13},x_{23}\right)=0,& \quad &
\bar{K}\left(x_{1},x_{2},x_{3},x_{123}\right)=0,\nonumber\\
B\left(x,x_{2},x_{3},x_{23}\right)=0,& \quad &
\bar{B}\left(x_{1},x_{12},x_{13},x_{123}\right)=0,\\
C\left(x,x_{1},x_{3},x_{13}\right)=0,& \quad &
\bar{C}\left(x_{2},x_{12},x_{23},x_{123}\right)=0,\nonumber
\end{eqnarray}
assigned to the faces of a 3D cube as on Figure~\ref{fig:cube2}, is 3D consistent and possesses the tetrahedron property expressed by the two equations
\[
A\left(x,x_{1},x_{2},x_{12}\right)=0,\qquad
\bar{A}\left(x_{3},x_{13},x_{23},x_{123}\right)=0.
\]
3D consistency of \eqref{dualsystem} is understood as the property of the initial value problem with the initial date $x$, $x_{3}$, $x_{13}$ and $x_{23}$.
\end{lemma}

Now we can formulate the main result of \cite{classification} concerning 3D consistent six-tuples containing equations of type \Hvier\  but no equations of type \Hsechs. In this theorem, we use normal forms $Q_k^\epsilon$ of equations of the list Q which are slightly different from those adopted in \cite{ABS1, ABS2}, namely
\begin{itemize}
\item[$Q_{1}^{\epsilon}$:] $Q=\alpha_{1}\left(x_{1}x_{2}+x_{3}x_{4}\right)-\alpha_{2}\left(x_{1}x_{4}+x_{2}x_{3}\right)-
    \left(\alpha_{1}-\alpha_{2}\right)\left(x_{1}x_{3}+x_{2}x_{4}\right)$\par$\qquad+
    \epsilon^{2}\alpha_{1}\alpha_{2}\left(\alpha_{1}-\alpha_{2}\right)$,\par
\item[$Q_{2}^{\epsilon}$:] $Q=\alpha_{1}\left(x_{1}x_{2}+x_{3}x_{4}\right)-\alpha_{2}\left(x_{1}x_{4}+x_{2}x_{3}\right)-
    \left(\alpha_{1}-\alpha_{2}\right)\left(x_{1}x_{3}+x_{2}x_{4}\right)+
    \alpha_{1}\alpha_{2}\left(\alpha_{1}-\alpha_{2}\right)$\par
    $\qquad+
    \epsilon\alpha_{1}\alpha_{2}\left(\alpha_{1}-\alpha_{2}\right)\left(x_{1}+x_{2}+x_{3}+x_{4}\right)
    -\epsilon^{2}\alpha_{1}\alpha_{2}\left(\alpha_{1}-\alpha_{2}\right)
    \left(\alpha_{1}^{2}-\alpha_{1}\alpha_{2}+\alpha_{2}^{2}\right)$,\par
\item[$Q_{3}^{\epsilon}$:] $Q=\sinh\left(2\alpha_{1}\right)\left(x_{1}x_{2}+x_{3}x_{4}\right)-
    \sinh\left(2\alpha_{2}\right)\left(x_{1}x_{4}+x_{2}x_{3}\right)$\par $\qquad-\sinh\left(2\left(\alpha_{1}-\alpha_{2}\right)\right)\left(x_{1}x_{3}+x_{2}x_{4}\right)
    -4\delta^{2}\epsilon^{2}\sinh\left(2\alpha_{1}\right)\sinh\left(2\alpha_{2}\right)
    \sinh\left(2\left(\alpha_{1}-\alpha_{2}\right)\right)$.
\end{itemize}
\begin{theo}\label{th:class}
Modulo independent M\"obius transformations of all fields, there are nine 3D consistent six-tuples of quad-equations with the tetrahedron property which contain equations of type \Hvier\  but no equations of type \Hsechs:
\begin{itemize}
\item[a)] Three systems with the biquadratics pattern as on Figure \ref{fig:5a}, with all six face equations being rhombic $H_k^\epsilon$, and the tetrahedron equations $Q_k^\epsilon$ (for $k=1,2,3$).
\item[b)] Three systems with the biquadratics pattern as on Figure \ref{fig:5b}, with two pairs of face equations being trapezoidal $H_k^\epsilon$, one pair of face equations $Q_k^\epsilon$, and the tetrahedron equations $H_k^\epsilon$ (for $k=1,2,3$).
\item[c)] Three systems with the biquadratics pattern as on Figure \ref{fig:5c}, with two pairs of face equations being trapezoidal $H_k^\epsilon$, the third pair of face equations being rhombic $H_k^\epsilon$, and both tetrahedron equations $H_k^\epsilon$ (for $k=1,2,3$).
\end{itemize}
\end{theo}

All these six-tuples can be embedded into the lattice $\Z^{3}$ by reflecting the cubes, see details of this construction in \cite{classification}. The Lagrangian formulation of these systems continued to the whole lattice constitutes the main subject of the present paper.

\section{Three-leg forms and Lagrangians} \label{threelegforms}

As already demonstrated in \cite{BS1}, the Lagrangian structures of quad-equations are closely related to their three-leg forms, which, in turn, can be derived from the biquadratics related to the quad-equations.
\begin{lemma}\label{lem:1param}
For any of the polynomials $Q_k^\epsilon$, $H_k^\epsilon$ $(k=1,2,3)$ there exits a re-scaling of the function $Q$ such that the re-scaled biquadratics depend only on one parameter, namely the parameter assigned to the corresponding edge:
\begin{eqnarray*}
Q^{1,2}(x_1,x_2;\alpha_1,\alpha_2) & = & \kappa(\alpha_1,\alpha_2)h^{1,2}(x_1,x_2;\alpha_1),\\
Q^{1,4}(x_1,x_4;\alpha_1,\alpha_2) & = & \kappa(\alpha_1,\alpha_2)h^{1,4}(x_1,x_4;\alpha_2),\\
Q^{1,3}(x_1,x_3;\alpha_1,\alpha_2) & = & \kappa(\alpha_1,\alpha_2)h^{1,3}(x_1,x_3;\alpha_1-\alpha_2).
\end{eqnarray*}
\end{lemma}
The proof consists of a direct verification with the scaling factors given in Sections \ref{TypeQ}, \ref{TypeH}. In particular, for the polynomials $Q_k^\epsilon$, due to their high symmetry, all three biquadratics are given essentially by one and the same non-degenerate symmetric polynomial $h(x,y;\alpha)=h(y,x;\alpha)$:
\begin{eqnarray*}
 & h^{1,2}(x_1,x_2;\alpha_1)=h(x_1,x_2;\alpha_1),\quad h^{1,4}(x_1,x_4;\alpha_2)=-h(x_1,x_4;\alpha_2), &\\
 & h^{1,3}(x_1,x_3;\alpha_1-\alpha_2)=-h(x_1,x_3;\alpha_1-\alpha_2). &
\end{eqnarray*}
For the polynomials $H_k^\epsilon$, one still has
\[
  h^{1,2}(x_1,x_2;\alpha_1)=h(x_1,x_2;\alpha_1),\quad h^{1,4}(x_1,x_4;\alpha_2)=-h(x_1,x_4;\alpha_2),
\]
where the biquadratic polynomial $h(x,y;\alpha)$ is degenerate and non-symmetric.

The signs on the right-hand sides of the latter equations are responsible for the choice of signs in the three-leg form of quad-equations:

\begin{lemma} \label{three-legs}
Equations $Q_k^{\epsilon}$ and rhombic equations $H_k^{\epsilon}$ possess three-leg forms
\[
\psi(x_{1},x_{2};\alpha_{1})-\psi(x_{1},x_{4};\alpha_{2})-\phi(x_{1},x_{3};\alpha_{1}-\alpha_{2})=0
\]
and
\[
\bar{\psi}(x_{2},x_{1};\alpha_{1})-\bar{\psi}(x_{2},x_{3};\alpha_{2})-
\bar{\phi}(x_{2},x_{4};\alpha_{1}-\alpha_{2})=0,
\]
centered at $x_1$ and at $x_2$, respectively, with the functions $\psi,\bar\psi,\phi$ and $\bar\phi$ listed in Sections~\ref{TypeQ}, \ref{TypeH}. For equations $Q_k^\epsilon$, one has $\psi=\bar\psi=\phi=\bar\phi$. If $x_i=f_i(X_i)$ denotes the uniformizing changes of variables for $\sqrt{r^i(x_i)}$, where $r^i(x_i)$ is the discriminant corresponding to the vertex $x_i$, then the leg functions satisfy the following relations:
\begin{eqnarray} \label{tl1}
\psi(x_{1},x_{2};\alpha_{1}) & = & \dfrac{\partial f_{1}(X_{1})}{\partial X_{1}}\int^{x_{2}}
  \dfrac{d\xi_{2}}{h^{1,2}(x_{1},\xi_{2};\alpha_{1})},\nonumber\\
\psi(x_{1},x_{4};\alpha_{2}) & = & -\dfrac{\partial f_{1}(X_{1})}{\partial X_{1}}\int^{x_{4}}
  \dfrac{d\xi_{4}}{h^{1,4}(x_{1},\xi_{4};\alpha_{2})}, \nonumber \\
\phi(x_{1},x_{3};\alpha_{1}-\alpha_{2}) & = &
  -\dfrac{\partial f_{1}(X_{1})}{\partial X_{1}}\int^{x_{3}}
  \dfrac{d\xi_{3}}{h^{1,3}(x_{1},\xi_{3};\alpha_{1}-\alpha_{2})},
\end{eqnarray}
and
\begin{eqnarray}
\bar{\psi}(x_{2},x_{1};\alpha_{1}) & = & \dfrac{\partial f_{2}(X_{2})}{\partial X_{2}}\int^{x_{1}}
  \dfrac{d\xi_{1}}{h^{2,1}(x_{2},\xi_{1};\alpha_{1})}, \nonumber\\
\bar{\psi}(x_{2},x_{3};\alpha_{2}) & = & -\dfrac{\partial f_{2}(X_{2})}{\partial X_{2}}\int^{x_{3}}
  \dfrac{d\xi_{3}}{h^{2,3}(x_{2},\xi_{3};\alpha_{2})}, \nonumber\\
\bar{\phi}(x_{2},x_{4};\alpha_{1}-\alpha_{2}) & = &
  -\dfrac{\partial f_{2}(X_{2})}{\partial X_{2}}\int^{x_{4}}
  \dfrac{d\xi_{4}}{h^{2,4}(x_{2},\xi_{4};\alpha_{1}-\alpha_{2})}.
\end{eqnarray}
In addition, we have the following properties:
\begin{itemize}
\item Symmetry:
\begin{eqnarray}\label{thl1}
\frac{\partial}{\partial X_{2}}\psi(x_{1},x_{2};\alpha_{1}) & = &
  \frac{\partial}{\partial X_{1}}\bar{\psi}(x_{2},x_{1};\alpha_{1}),\\
\frac{\partial}{\partial X_{3}}\phi(x_{1},x_{3};\alpha_{1}-\alpha_{2}) & = &
  \frac{\partial}{\partial X_{1}}\phi(x_{3},x_{1};\alpha_{1}-\alpha_{2}),\\
\frac{\partial}{\partial X_{4}}\bar{\phi}(x_{2},x_{4};\alpha_{1}-\alpha_{2}) & = &
  \frac{\partial}{\partial X_{2}}\bar{\phi}(x_{4},x_{2};\alpha_{1}-\alpha_{2}).
\end{eqnarray}

\item Derivative with respect to the parameter:
\begin{eqnarray}\label{tl2}
\frac{\partial}{\partial \alpha_{1}}\psi(x_{1},x_{2};\alpha_{1}) & = &
  \frac{\partial}{\partial X_{1}}\log h^{1,2}(x_{1},x_{2};\alpha_{1}),\\
\frac{\partial}{\partial \alpha_{1}}\bar{\psi}(x_{2},x_{1};\alpha_{1}) & = &
  \frac{\partial}{\partial X_{2}}\log h^{1,2}(x_{1},x_{2};\alpha_{1}),\\
\frac{\partial}{\partial \alpha_{1}}\phi(x_{1},x_{3};\alpha_{1}-\alpha_{2}) & = &
  -\frac{\partial}{\partial X_{1}}\log h^{1,3}(x_{1},x_{3};\alpha_{1}-\alpha_{2}),\\
\frac{\partial}{\partial \alpha_{1}}\bar{\phi}(x_{2},x_{4};\alpha_{1}-\alpha_{2}) & = &
  -\frac{\partial}{\partial X_{2}}\log h^{2,4}(x_{2},x_{4};\alpha_{1}-\alpha_{2}).
\end{eqnarray}
\end{itemize}
\end{lemma}
\begin{proof}
Similar statements for equations of the ABS list were already given in \cite{ABS1, BS1}. Here we provide the reader with slightly more explanations why these statements hold true. Namely, the first statement follows from the result (due to V. Adler and published as Exercise 6.16 in \cite[p.281]{DDG}) that the three-leg form of an arbitrary multi-affine equation $Q=0$, centered at $x_1$, reads
\[
\int^{x_{2}}\frac{d\xi_2}{Q^{1,2}(x_{1},\xi_2)}+\int^{x_{3}}\frac{d\xi_3}{Q^{1,3}(x_{1},\xi_3)}+
\int^{x_{4}}\frac{d\xi_4}{Q^{1,4}(x_{1},\xi_4)}=0.
\]
Due to Lemma \ref{lem:1param}, for equations $H_k^\epsilon$, $Q_k^\epsilon$ one can replace here $Q^{i,j}$ by the biquadratics $h^{i,j}$ depending on one parameter each. Furthermore, all three leg functions in the latter formula can be simultaneously multiplied by an arbitrary function of $x_1$. The choice of this function as $\partial f_1(X_1)/\partial X_1$ ensures the symmetry properties of the derivative of the leg function with respect to the second argument:
\[
\frac{\partial}{\partial X_{2}}\psi(x_{1},x_{2};\alpha_{1})=\frac{\partial f_{1}(X_{1})}{\partial X_{1}}
\cdot\frac{\partial f_{2}(X_{2})}{\partial X_{2}}\cdot \frac{1}{h^{1,2}(x_{1},x_{2};\alpha_{1})}
=\frac{\partial}{\partial X_{1}}\bar{\psi}\left(x_{2},x_{1};\alpha_{1}\right).
\]
The both diagonal biquadratics are symmetric but different, so that one has two different long leg functions $\phi$ and $\bar\phi$. The formulas for the derivatives of the leg functions with respect to parameters are checked by a direct computation.
\end{proof}

\begin{rem}
For trapezoidal equations $H_k^\epsilon$, the notion of short and long legs might be somewhat misleading. We will assume that, for both the rhombic and the trapezoidal versions of $H_k^\epsilon$, the functions $\phi$ and $\bar\phi$ correspond to two non-degenerate biquadratics (black and white, respectively), while the functions $\psi$ and $\bar\psi$ correspond to four degenerate biquadratics, with the first argument (the center point of the three-leg form) being black, resp. white. Thus, the three-leg forms of a trapezoidal equation $H_k^\epsilon$, as depicted on Figure \ref{fig:2b}, read:
\[
\phi(x_1,x_2;\alpha_1)-\psi(x_1,x_4;\alpha_2)+\psi(x_1,x_3;\alpha_2-\alpha_1)=0
\]
(centered at a black vertex $x_1$), resp.
\[
\bar\phi(x_3,x_4;\alpha_1)-\bar\psi(x_3,x_2;\alpha_2)+\bar\psi(x_3,x_1;\alpha_2-\alpha_1)=0
\]
(centered at a white vertex $x_3$).
\end{rem}

\begin{rem}
We recall that the three-leg forms of quad-equations yield the so called Laplace type equations of the stars of the vertices of the underlying quad-graph. These Laplace type equations say that the sum of long-leg functions for diagonals of all quadrilaterals adjacent to a vertex is equal to zero. For the rhombic equations $H_k^\epsilon$, the long-leg functions $\phi,\bar\phi$ are the same as for the type Q equations, so that they do not lead to new Laplace type equations. For the trapezoidal equations $H_k^\epsilon$, the long-leg functions are $\psi,\bar\psi$, so that one arrives at novel Laplace type equations. For instance, the Laplace type equations at the vertices $x_3$ (black) and $x_6$ (white) on Figure~\ref{fig:8} read:
\begin{figure}[htbp]
   \centering
  \includegraphics[scale=1.0]{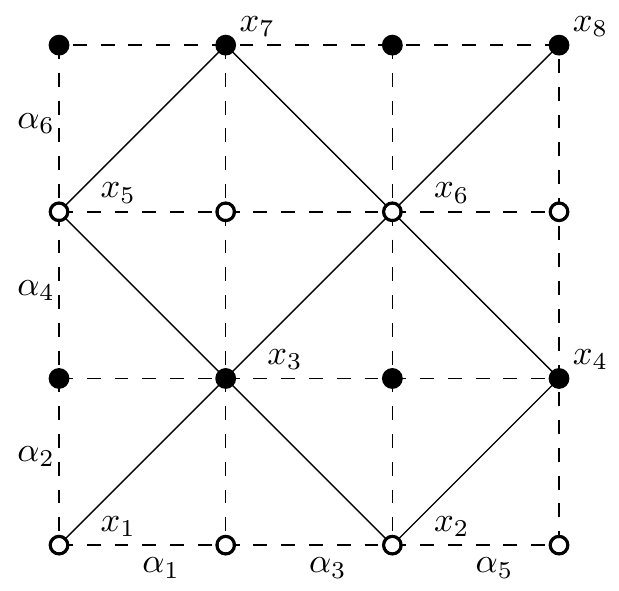}
   \caption{Laplace type equations for a $\mathbb Z^2$ system of trapezoidal equations $H_k^\epsilon$}
   \label{fig:8}
\end{figure}
\[
\psi(x_3,x_6;\alpha_4-\alpha_3)-\psi(x_3,x_5;\alpha_4-\alpha_1)+\psi(x_3,x_1;\alpha_2-\alpha_1)
   -\psi(x_3,x_2;\alpha_2-\alpha_3)=0,
\]
resp.
\[
\bar\psi(x_6,x_8;\alpha_6-\alpha_5)-\bar\psi(x_6,x_7;\alpha_6-\alpha_3)
+\bar\psi(x_6,x_3;\alpha_4-\alpha_3)-\bar\psi(x_6,x_4;\alpha_4-\alpha_5)=0.
\]
\end{rem}

We are now in a position to construct functions which will play the role of Lagrangians in the action functionals for system of quad-equations.

\begin{lemma} \label{lagr}
For equations $Q_k^{\epsilon}$, $H_k^{\epsilon}$, there exist functions $L\left(X_{1},X_{2};\alpha_{1}\right)$ (defined on directed edges starting at a black vertex and ending at a white vertex),
\[
\Lambda\left(X_{1},X_{3};\alpha_{1}-\alpha_{2}\right)=
\Lambda\left(X_{3},X_{1};\alpha_{1}-\alpha_{2}\right)
 \]
 (defined on edges connecting two black vertices), and
 \[
 \bar{\Lambda}\left(X_{2},X_{4};\alpha_{1}-\alpha_{2}\right)=
\bar{\Lambda}\left(X_{4},X_{2};\alpha_{1}-\alpha_{2}\right)
 \]
 (defined on edges connecting two white vertices) such that, upon changes of variables $x_{i}=f_{i}\left(X_{i}\right)$ introduced in Lemma \ref{three-legs}, there hold the following relations:
\begin{eqnarray}
& \psi\left(x_{1},x_{2};\alpha_{1}\right)=\dfrac{\partial}{\partial X_{1}}L\left(X_{1},X_{2};\alpha_{1}\right),\quad
\bar{\psi}\left(x_{2},x_{1};\alpha_{1}\right)=\dfrac{\partial}{\partial X_{2}}L\left(X_{1},X_{2};\alpha_{1}\right), & \\
& \phi\left(x_{1},x_{3};\alpha_{1}-\alpha_{2}\right)=-\dfrac{\partial}{\partial X_{1}}\Lambda\left(X_{1},X_{3};\alpha_{1}-\alpha_{2}\right), & \\
& \bar{\phi}\left(x_{2},x_{4};\alpha_{1}-\alpha_{2}\right)=-\dfrac{\partial}{\partial X_{2}}\bar{\Lambda}\left(X_{2},X_{4};\alpha_{1}-\alpha_{2}\right). &
\end{eqnarray}
These functions (defined up to arbitrary additive functions of the respective parameters) can be chosen so that
\begin{eqnarray}
\frac{\partial}{\partial\alpha_{1}}L\left(X_{1},X_{2};\alpha_{1}\right) & = &
\log h^{1,2}\left(x_{1},x_{2};\alpha_{1}\right),\\
\frac{\partial}{\partial\alpha_{1}}\Lambda\left(X_{1},X_{3};\alpha_{1}-\alpha_{2}\right) & = &
\log h^{1,3}\left(x_{1},x_{3};\alpha_{1}-\alpha_{2}\right),\\
\frac{\partial}{\partial \alpha_{1}}\bar{\Lambda}\left(X_{2},X_{4};\alpha_{1}-\alpha_{2}\right) & = &
\log h^{2,4}\left(x_{2},x_{4};\alpha_{1}-\alpha_{2}\right).
\end{eqnarray}
\end{lemma}

The main theorem concerning the Lagrangians for a single quad-equation, which serves as the major ingredient of the proof of the flip invariance of the action functionals,  reads:
\begin{theo} \label{th:single}
For equations $H_{k}^{\epsilon}$, the Lagrangians $L$, $\Lambda$ and $\bar{\Lambda}$ can be chosen so that the following relation holds on solutions of the equation $Q(x_1,x_2,x_3,x_4)=0$:

a) rhombic case with the enumeration of fields as on Figure~\ref{fig:2a}:
\begin{multline} \label{singlequad rhomb} L\left(X_{1},X_{2};\alpha_{1}\right)+L\left(X_{3},X_{4};\alpha_{1}\right)-
L\left(X_{1},X_{4};\alpha_{2}\right)-L\left(X_{3},X_{2};\alpha_{2}\right)\\
-\Lambda\left(X_{1},X_{3};\alpha_{1}-\alpha_{2}\right)-
\bar{\Lambda}\left(X_{2},X_{4};\alpha_{1}-\alpha_{2}\right)=0;
\end{multline}

b) trapezoidal case with the enumeration of fields as on Figure~\ref{fig:2b}:
\begin{multline} \label{singlequad trapez}
\Lambda(X_1,X_2;\alpha_1)+\bar{\Lambda}(X_3,X_4;\alpha_1)
-L(X_1,X_4;\alpha_2)-L(X_2,X_3;\alpha_2)\\
+L(X_1,X_3;\alpha_2-\alpha_1)+L(X_2,X_4;\alpha_2-\alpha_1)=0.
\end{multline}

For equations $Q_k^\epsilon$, the analogous statement holds true if one sets all Lagrangians $L,\Lambda$, and $\bar\Lambda$ equal (to $\Lambda$, say).
\end{theo}
\begin{proof}
The proof of this theorem is completely analogous to the one of Theorem~1 in \cite{BS1} and is obtained by showing that the gradient of the function on the left-hand side vanishes on the solutions of $Q=0$ due to the equivalence of the equation $Q=0$ to its three-leg forms.
\end{proof}

It is important to notice that in a 3D consistent six-tuple of quad-equations possessing the tetrahedron property, the choice of Lagrangians of every quad-equation (be it the face equation or the tetrahedron equation) can be made in a consistent manner, i.e., on every edge (or diagonal) the two Lagrangians coming from the two different equations of the quadrilaterals sharing this edge coincide. This follows from the next result which was given for equations of the ABS list in \cite{ABS2}, and was proven in full generality in \cite{classification}):
\begin{lemma} \label{biquads}
Consider a 3D consistent six-tuple~\eqref{system} of quad-equations possessing the tetrahedron property. Then:
\begin{enumerate}
\item For any edge of the cube, the two biquadratics corresponding to this edge coincide up to a constant factor; for example
\[
\frac{A^{0,1}}{C^{0,1}}={\rm const},\qquad
\frac{B^{0,2}}{A^{0,2}}={\rm const},\qquad
\frac{C^{0,3}}{B^{0,3}}={\rm const}.
\]
\item The product of these factors around one vertex is equal to $-1$; for example, for the vertex $x$ one has:
\[
\frac{A^{0,1}}{C^{0,1}}\cdot\frac{B^{0,2}}{A^{0,2}}\cdot\frac{C^{0,3}}{B^{0,3}}=-1.
\]
One can choose the normalization of the polynomials $A,\ldots,\bar C$ so that all such factors (for all 12 edges of the cube) are equal to $-1$.
\end{enumerate}
\end{lemma}
As a consequence of this result, for consistent systems of quad-equations, each edge $\mathbb Z^m$ carries a well defined leg function and a well defined Lagrangian, in the sense that these objects do not depend on the elementary square of $\mathbb Z^m$ adjacent to this edge. Moreover, due to the construction of Lemma \ref{tetrahedronUse}, the same is true for diagonals of elementary squares of $\mathbb Z^m$, which support leg functions and Lagrangians coming alternatively from edges of black and white tetrahedra supporting their own quad-equations.

\section{Flip invariance of action functional} \label{flipQuad}

Following the main idea of \cite{LN}, we consider the functional for a multidimensionally consistent system of quad-equations defined on any quad-surfaces $\Sigma$ in the $\mathbb Z^m$ by the formula
\begin{equation}
S=\sum_{\sigma_{ij}\in\Sigma}\Ell(\sigma_{ij}),
\end{equation}
where $\Ell(\sigma_{ij})$ is the value of a discrete Lagrangian two-form on an elementary (oriented) square $\sigma_{ij}=(n,n+e_{i},n+e_{i}+e_{j},n+e_{j})$,
\begin{equation*}\label{func2}
\Ell\left(\sigma_{ij}\right)=\Ell\left(X,X_{i},X_{j};\alpha_{i},\alpha_{j}\right).
\end{equation*}
For equations $H_k^\epsilon$ and $Q_k^\epsilon$ defined on black-and-white (but not necessarily bipartite) quad-graphs, one can encounter eight types of elementary squares depicted on Figure~\ref{fig:9}.
\begin{figure}[htbp]
   \centering
   \subfloat[]{\label{fig:9a}\includegraphics[scale=1.0]{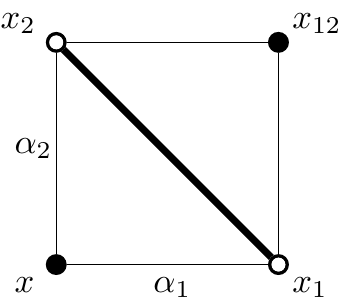}}\qquad
   \subfloat[]{\label{fig:9b}\includegraphics[scale=1.0]{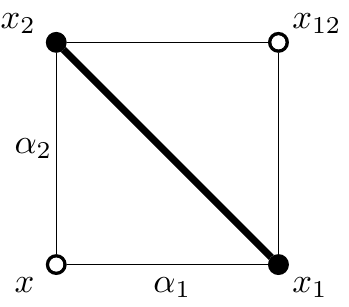}}\qquad
   \subfloat[]{\label{fig:9c}\includegraphics[scale=1.0]{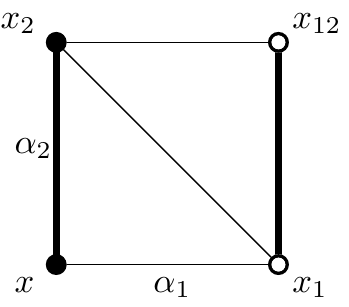}}\\
   \subfloat[]{\label{fig:9d}\includegraphics[scale=1.0]{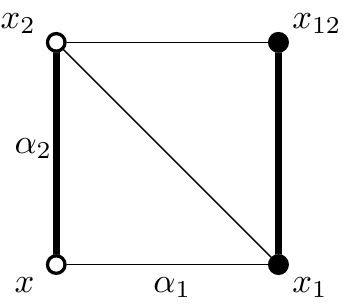}}\qquad
   \subfloat[]{\label{fig:9e}\includegraphics[scale=1.0]{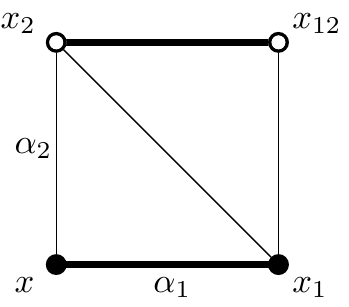}}\qquad
   \subfloat[]{\label{fig:9f}\includegraphics[scale=1.0]{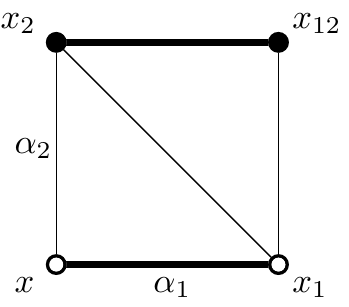}}\\
   \subfloat[]{\label{fig:9g}\includegraphics[scale=1.0]{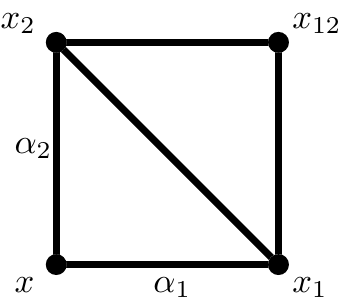}}\qquad
   \subfloat[]{\label{fig:9h}\includegraphics[scale=1.0]{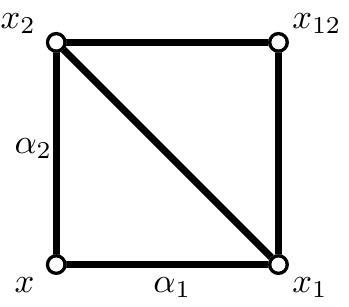}}
   \caption{Eight types of elementary squares of a black-and-white lattice}
   \label{fig:9}
\end{figure}

The corresponding values of the discrete two-form are:
\begin{align}
\Ell(X,X_1,X_2;\alpha_1,\alpha_2) &= L(X,X_1;\alpha_1)-L(X,X_2;\alpha_2)-\bar\Lambda(X_1,X_2,\alpha_1-\alpha_2); \tag{a} \\
\Ell(X,X_1,X_2;\alpha_1,\alpha_2) &= L(X_1,X;\alpha_1)-L(X_2,X;\alpha_2)-\Lambda(X_1,X_2,\alpha_1-\alpha_2); \tag{b} \\
\Ell(X,X_1,X_2;\alpha_1,\alpha_2) &= L(X,X_1;\alpha_1)-\Lambda(X,X_2;\alpha_2)-L(X_2,X_1,\alpha_1-\alpha_2); \tag{c} \\
\Ell(X,X_1,X_2;\alpha_1,\alpha_2) &= L(X_1,X;\alpha_1)-\bar\Lambda(X,X_2;\alpha_2)-L(X_1,X_2,\alpha_1-\alpha_2); \tag{d} \\
\Ell(X,X_1,X_2;\alpha_1,\alpha_2) &= \Lambda(X,X_1;\alpha_1)-L(X,X_2;\alpha_2)+L(X_1,X_2,\alpha_2-\alpha_1); \tag{e} \\
\Ell(X,X_1,X_2;\alpha_1,\alpha_2) &= \bar\Lambda(X,X_1;\alpha_1)-L(X_2,X;\alpha_2)+L(X_2,X_1,\alpha_2-\alpha_1); \tag{f} \\
\Ell(X,X_1,X_2;\alpha_1,\alpha_2) &= \Lambda(X,X_1;\alpha_1)-\Lambda(X,X_2;\alpha_2)-\Lambda(X_1,X_2,\alpha_1-\alpha_2); \tag{g} \\
\Ell(X,X_1,X_2;\alpha_1,\alpha_2) &= \bar\Lambda(X,X_1;\alpha_1)-\bar\Lambda(X,X_2;\alpha_2)-\bar\Lambda(X_1,X_2,\alpha_1-\alpha_2). \tag{h}
\end{align}
The choice of the sign of the third term in the expressions (c)--(f) is dictated by the requirement that the two terms with the Lagrangian $L$ have opposite signs.

With these definitions, one has the value of the action functional on any quad-surface in $\mathbb Z^m$, including $\Z^2$. In the latter case the action can be represented as the sum of Lagrangians over all edges of $\Z^2$ and over all diagonals running from north-west to south-east (or in the opposite direction, since they could be directed). The term ``action functional'' is justified by the following statement.
\begin{theo}\label{quad1}
For any equation $H_k^\epsilon$ (rhombic or trapezoidal) or $Q_k^\epsilon$ on the regular square lattice $\Z^{2}$, the solutions are critical points of the action functional.
\end{theo}
Our main result reads as follows:
\begin{theo}
For any 3D consistent system consisting of quad-equations $H_k^\epsilon$ and/or $Q_k^\epsilon$ (listed in Theorem \ref{th:class}), the action functional is flip-invariant in the sense that the following relation is satisfied on solutions:
\begin{equation}\label{flip inv}
\Delta_{1}\Ell\left(X,X_{2},X_{3};\alpha_{2},\alpha_{3}\right)+
\Delta_{2}\Ell\left(X,X_{3},X_{1};\alpha_{3},\alpha_{1}\right)+
\Delta_{3}\Ell\left(X,X_{1},X_{2};\alpha_{1},\alpha_{2}\right)=0.
\end{equation}
Here $\Delta_{i}$ denotes the difference operator that acts on vertex functions according to $$\Delta_{i}f\left(x\right)=f\left(x_{i}\right)-f\left(x\right).$$
\end{theo}
\begin{proof} We give the proof which generalizes the first proof of the corresponding result for equations from the ABS list given in \cite{BS1}. One adds the formulae from Theorem~\ref{th:single} for the three faces and one tetrahedron adjacent to the point $x$. We give details for the three elementary cubes illustrated on Figure~\ref{fig:5} and omit the details for the cubes obtained by the color inversion and rotations.

{\em Six-tuple from Figure~\ref{fig:5a}:} Add the four relations
\begin{multline*}
L\left(X,X_1;\alpha_{1}\right)+L\left(X_{12},X_{2};\alpha_{1}\right)-
L\left(X,X_{2};\alpha_{2}\right)-L\left(X_{12},X_{1};\alpha_{2}\right)\\
-\Lambda\left(X,X_{12};\alpha_{1}-\alpha_{2}\right)-
\bar{\Lambda}\left(X_{1},X_{2};\alpha_{1}-\alpha_{2}\right)=0,
\end{multline*}
\begin{multline*}
-L\left(X,X_1;\alpha_{1}\right)-L\left(X_{13},X_{3};\alpha_{1}\right)+
L\left(X,X_{3};\alpha_{3}\right)+L\left(X_{13},X_{1};\alpha_{3}\right)\\
+\Lambda\left(X,X_{13};\alpha_{1}-\alpha_{3}\right)+
\bar{\Lambda}\left(X_{1},X_{3};\alpha_{1}-\alpha_{3}\right)=0,
\end{multline*}
\begin{multline*}
L\left(X,X_2;\alpha_{2}\right)+L\left(X_{23},X_{3};\alpha_{2}\right)-
L\left(X,X_{3};\alpha_{3}\right)-L\left(X_{23},X_{2};\alpha_{3}\right)\\
-\Lambda\left(X,X_{23};\alpha_{2}-\alpha_{3}\right)-
\bar{\Lambda}\left(X_{2},X_{3};\alpha_{2}-\alpha_{3}\right)=0,
\end{multline*}
\begin{multline*}
\Lambda\left(X,X_{23};\alpha_{2}-\alpha_{3}\right)+
\Lambda\left(X_{12},X_{13};\alpha_{2}-\alpha_{3}\right)-
\Lambda\left(X,X_{13};\alpha_{1}-\alpha_{3}\right)-
\Lambda\left(X_{12},X_{23};\alpha_{1}-\alpha_{3}\right)\\
+\Lambda\left(X,X_{12};\alpha_{1}-\alpha_{2}\right)+
\Lambda\left(X_{13},X_{23};\alpha_{1}-\alpha_{2}\right)=0.
\end{multline*}
The result reads:
\begin{eqnarray*}
\lefteqn{L\left(X_{12},X_{1};\alpha_{2}\right)
+L\left(X_{23},X_{2};\alpha_{3}\right)
+L\left(X_{13},X_{3};\alpha_{1}\right)}\\
&& -L\left(X_{13},X_{1};\alpha_{3}\right)
-L\left(X_{12},X_{2};\alpha_{1}\right)
-L\left(X_{23},X_{3};\alpha_{2}\right)\\
&& -\Lambda\left(X_{12},X_{13};\alpha_{2}-\alpha_{3}\right)
-\Lambda\left(X_{12},X_{23};\alpha_{3}-\alpha_{1}\right)
-\Lambda\left(X_{13},X_{23};\alpha_{1}-\alpha_{2}\right)\\
&& +\bar{\Lambda}\left(X_{2},X_{3};\alpha_{2}-\alpha_{3}\right)
+\bar{\Lambda}\left(X_{1},X_{3};\alpha_{3}-\alpha_{1}\right)
+\bar{\Lambda}\left(X_{1},X_{2};\alpha_{1}-\alpha_{2}\right)=0.
\end{eqnarray*}
which is the flip invariance relation (\ref{flip inv}) in this case.

{\em Six-tuple from Figure~\ref{fig:5b}:} Add the four relations
\begin{multline*}
\Lambda\left(X,X_1;\alpha_{1}\right)+\Lambda\left(X_{2},X_{12};\alpha_{1}\right)-
\Lambda\left(X,X_{2};\alpha_{2}\right)-\Lambda\left(X_{1},X_{12};\alpha_{2}\right)\\
-\Lambda\left(X,X_{12};\alpha_{1}-\alpha_{2}\right)-
\Lambda\left(X_{1},X_{2};\alpha_{1}-\alpha_{2}\right)=0,
\end{multline*}
\begin{multline*}
-\Lambda\left(X,X_1;\alpha_{1}\right)-\bar\Lambda\left(X_{3},X_{13};\alpha_{1}\right)+
L\left(X,X_{3};\alpha_{3}\right)+L\left(X_{1},X_{13};\alpha_{3}\right)\\
+L\left(X,X_{13};\alpha_{1}-\alpha_{3}\right)+
L\left(X_{1},X_{3};\alpha_{1}-\alpha_{3}\right)=0,
\end{multline*}
\begin{multline*}
\Lambda\left(X,X_2;\alpha_{2}\right)+\bar\Lambda\left(X_{3},X_{23};\alpha_{2}\right)-
L\left(X,X_{3};\alpha_{3}\right)-L\left(X_{2},X_{23};\alpha_{3}\right)\\
-L\left(X,X_{23};\alpha_{2}-\alpha_{3}\right)-
L\left(X_{2},X_{3};\alpha_{2}-\alpha_{3}\right)=0,
\end{multline*}
\begin{multline*}
L\left(X,X_{23};\alpha_{2}-\alpha_{3}\right)+
L\left(X_{12},X_{13};\alpha_{2}-\alpha_{3}\right)-
L\left(X,X_{13};\alpha_{1}-\alpha_{3}\right)-
L\left(X_{12},X_{23};\alpha_{1}-\alpha_{3}\right)\\
+\Lambda\left(X,X_{12};\alpha_{1}-\alpha_{2}\right)+
\bar\Lambda\left(X_{13},X_{23};\alpha_{1}-\alpha_{2}\right)=0.
\end{multline*}
The result reads:
\begin{eqnarray*}
\lefteqn{\Lambda\left(X_{1},X_{12};\alpha_{2}\right)
+L\left(X_{2},X_{23};\alpha_{3}\right)
+\bar\Lambda\left(X_{3},X_{13};\alpha_{1}\right)}\\
&& -L\left(X_{1},X_{13};\alpha_{3}\right)
-\Lambda\left(X_{2},X_{12};\alpha_{1}\right)
-\bar\Lambda\left(X_{3},X_{23};\alpha_{2}\right)\\
&& -L\left(X_{12},X_{13};\alpha_{2}-\alpha_{3}\right)
+L\left(X_{12},X_{23};\alpha_{1}-\alpha_{3}\right)
-\bar\Lambda\left(X_{13},X_{23};\alpha_{1}-\alpha_{2}\right)\\
&& +L\left(X_{2},X_{3};\alpha_{2}-\alpha_{3}\right)
-L\left(X_{1},X_{3};\alpha_{1}-\alpha_{3}\right)
+\Lambda\left(X_{1},X_{2};\alpha_{1}-\alpha_{2}\right)=0.
\end{eqnarray*}
which is the flip invariance relation (\ref{flip inv}) in this case.

{\em Six-tuple from Figure~\ref{fig:5c}:} Add the four relations
\begin{multline*}
L\left(X,X_1;\alpha_{1}\right)+L\left(X_{12},X_{2};\alpha_{1}\right)-
L\left(X,X_{2};\alpha_{2}\right)-L\left(X_{12},X_{1};\alpha_{2}\right)\\
-\Lambda\left(X,X_{12};\alpha_{1}-\alpha_{2}\right)-
\bar\Lambda\left(X_{1},X_{2};\alpha_{1}-\alpha_{2}\right)=0,
\end{multline*}
\begin{multline*}
-L\left(X,X_1;\alpha_{1}\right)-L\left(X_{3},X_{13};\alpha_{1}\right)+
\Lambda\left(X,X_{3};\alpha_{3}\right)+
\bar\Lambda\left(X_{1},X_{13};\alpha_{3}\right)\\
+L\left(X,X_{13};\alpha_{1}-\alpha_{3}\right)+
L\left(X_{3},X_{1};\alpha_{1}-\alpha_{3}\right)=0,
\end{multline*}
\begin{multline*}
L\left(X,X_2;\alpha_{2}\right)+L\left(X_{3},X_{23};\alpha_{2}\right)-
\Lambda\left(X,X_{3};\alpha_{3}\right)-\bar\Lambda\left(X_{2},X_{23};\alpha_{3}\right)\\
-L\left(X,X_{23};\alpha_{2}-\alpha_{3}\right)-
L\left(X_{3},X_{2};\alpha_{2}-\alpha_{3}\right)=0,
\end{multline*}
\begin{multline*}
L\left(X,X_{23};\alpha_{2}-\alpha_{3}\right)+
L\left(X_{12},X_{13};\alpha_{2}-\alpha_{3}\right)-
L\left(X,X_{13};\alpha_{1}-\alpha_{3}\right)-
L\left(X_{12},X_{23};\alpha_{1}-\alpha_{3}\right)\\
+\Lambda\left(X,X_{12};\alpha_{1}-\alpha_{2}\right)+
\bar\Lambda\left(X_{13},X_{23};\alpha_{1}-\alpha_{2}\right)=0.
\end{multline*}
The result reads:
\begin{eqnarray*}
\lefteqn{L\left(X_{12},X_{1};\alpha_{2}\right)
+\Lambda\left(X_{2},X_{23};\alpha_{3}\right)
+L\left(X_{3},X_{13};\alpha_{1}\right)}\\
&& -\Lambda\left(X_{1},X_{13};\alpha_{3}\right)
-L\left(X_{12},X_{2};\alpha_{1}\right)
-L\left(X_{3},X_{23};\alpha_{2}\right)\\
&& -L\left(X_{12},X_{13};\alpha_{2}-\alpha_{3}\right)
+L\left(X_{12},X_{23};\alpha_{1}-\alpha_{3}\right)
-\bar\Lambda\left(X_{13},X_{23};\alpha_{1}-\alpha_{2}\right)\\
&& +L\left(X_{3},X_{2};\alpha_{2}-\alpha_{3}\right)
-L\left(X_{3},X_{1};\alpha_{1}-\alpha_{3}\right)
+\bar\Lambda\left(X_{1},X_{2};\alpha_{1}-\alpha_{2}\right)=0.
\end{eqnarray*}
which is the flip invariance relation (\ref{flip inv}) in this case.
\end{proof}

\section*{Acknowledgments}
The first author is supported by the Berlin Mathematical School, the second author is partially supported by the DFG Research Unit 565 ``Polyhedral Surfaces''.

\appendix
\section{List of Quad-Equations}
\subsection{Type Q} \label{TypeQ}
\begin{itemize}
\item[$Q_{1}^{\epsilon}$:] $Q=\alpha_{1}\left(x_{1}x_{2}+x_{3}x_{4}\right)-\alpha_{2}\left(x_{1}x_{4}+x_{2}x_{3}\right)-
    \left(\alpha_{1}-\alpha_{2}\right)\left(x_{1}x_{3}+x_{2}x_{4}\right)$\par
    $\qquad+\epsilon^{2}\alpha_{1}\alpha_{2}\left(\alpha_{1}-\alpha_{2}\right)$,\par
$x_{i}=X_{i}$,\par
$\kappa(\alpha_1,\alpha_2)=2\alpha_1\alpha_2(\alpha_1-\alpha_2)$,\par
$h^{1,2}\left(x_{1},x_{2};\alpha_{1}\right)=
\dfrac{1}{2}\left(\dfrac{\left(x_{1}-x_{2}\right)^{2}}{\alpha_{1}}-\epsilon^{2}\alpha_{1}\right)$,\par
$\phi\left(x_{1},x_{2};\alpha_{1}\right)=\begin{cases}
\dfrac{1}{\epsilon}\log\dfrac{X_{1}-X_{2}+\epsilon \alpha_{1}}{X_{1}-X_{2}-\epsilon\alpha_{1}}, & \text{if } \epsilon\neq 0, \\
\dfrac{2\alpha_1}{X_{1}-X_{2}}, & \text{if } \epsilon=0 ;
\end{cases}$

\item[$Q_{2}^{\epsilon}$:] $Q=\alpha_{1}\left(x_{1}x_{2}+x_{3}x_{4}\right)-
    \alpha_{2}\left(x_{1}x_{4}+x_{2}x_{3}\right)-\left(\alpha_{1}-
    \alpha_{2}\right)\left(x_{1}x_{3}+x_{2}x_{4}\right)$\par $\qquad+\alpha_{1}\alpha_{2}\left(\alpha_{1}-\alpha_{2}\right)+
    \epsilon\alpha_{1}\alpha_{2}\left(\alpha_{1}-\alpha_{2}\right)\left(x_{1}+x_{2}+x_{3}+x_{4}\right)$\par $\qquad-\epsilon^{2}\alpha_{1}\alpha_{2}\left(\alpha_{1}-\alpha_{2}\right)
    \left(\alpha_{1}^{2}-\alpha_{1}\alpha_{2}+\alpha_{2}^{2}\right)$,\par
$x_{i}=\begin{cases}X_i(1+\epsilon X_{i}), &\text{if }\epsilon\neq 0,\\X_{i},&\text{if }\epsilon=0,\end{cases}$\par
$\kappa(\alpha_1,\alpha_2)=2\alpha_1\alpha_2(\alpha_1-\alpha_2)$,\par
$h^{1,2}\left(x_{1},x_{2};\alpha_{1}\right)=\dfrac{1}{2\alpha_{1}}
\left(\left(x_{1}-x_{2}\right)^{2}-\alpha_{1}^{2}-2\epsilon\alpha_{1}^{2}\left(x_{1}+x_{2}\right)
+\epsilon^{2}\alpha_{1}^{4}\right)$,\par
$\phi\left(x_{1},x_{2};\alpha_{1}\right)=\begin{cases}
\log\dfrac{\left(X_{1}-X_{2}+\alpha_{1}\right)
\left(1+\epsilon(X_{1}+X_{2}+\alpha_{1})\right)}{\left(X_{1}-X_{2}-\alpha_{1}\right)
\left(1+\epsilon(X_{1}+X_{2}-\alpha_{1})\right)}, & \text{if } \epsilon\neq 0,\\
\log\dfrac{X_{1}-X_{2}+\alpha_{1}}{X_{1}-X_{2}-\alpha_{1}}, & \text{if } \epsilon=0;
\end{cases}$

\item[$Q_{3}^{\epsilon}$:] $Q=\sinh(2\alpha_{1})\left(x_{1}x_{2}+x_{3}x_{4}\right)-
    \sinh(2\alpha_{2})\left(x_{1}x_{4}+x_{2}x_{3}\right)$\par $\qquad-\sinh(2(\alpha_{1}-\alpha_{2}))\left(x_{1}x_{3}+x_{2}x_{4}\right)$\par $\qquad-4\delta^{2}\epsilon^{2}\sinh(2\alpha_{1})\sinh(2\alpha_{2})
    \sinh(2(\alpha_{1}-\alpha_{2}))$,\par
$x_{i}=\begin{cases}2\delta\epsilon\cosh\left(2X_{i}\right), & \text{if }
\epsilon\neq 0, \\ \exp(2X_{i}), & \text{if }\epsilon=0,\end{cases}$\par
$\kappa(\alpha_1,\alpha_2)=\sinh(2\alpha_1)\sinh(2\alpha_2)\sinh(2(\alpha_1-\alpha_2))$, \par
$h^{1,2}\left(x_{1},x_{2};\alpha_{1}\right)=
\dfrac{\left(e^{\alpha_{1}}x_{1}-e^{-\alpha_{1}}x_{2}\right)
\left(e^{-\alpha_{1}}x_{1}-e^{2\alpha_{1}}x_{2}\right)}
{\sinh(2\alpha_{1})}+4\delta^{2}\epsilon^{2}\sinh(2\alpha_{1})$,\par
$\phi\left(x_{1},x_{2};\alpha_{1}\right)=
\begin{cases}
\log\dfrac{\sinh\left(X_{1}+X_{2}+\alpha_{1}\right)
\sinh\left(X_{1}-X_{2}+\alpha_{1}\right)}{\sinh\left(X_{1}+X_{2}-\alpha_{1}\right)
\sinh\left(X_{1}-X_{2}-\alpha_{1}\right)}, & \text{if } \epsilon\neq0,\\
\log\dfrac{\sinh\left(X_{1}-X_{2}+\alpha_{1}\right)}{\sinh\left(X_{1}-X_{2}-\alpha_{1}\right)}+2\alpha_1, & \text{if } \epsilon=0.
\end{cases}$
\end{itemize}

\subsection{Type \texorpdfstring{\Hvier}{H4}} \label{TypeH}
\begin{itemize}
\item[$H_{1}^{\epsilon}$:] $Q=\left(x_{1}-x_{3}\right)\left(x_{2}-x_{4}\right)+
    \left(\alpha_{2}-\alpha_{1}\right)\left(1-\epsilon^{2} x_{2}x_{4}\right)$,\par
$x_{i}=X_{i}$,\par
$\kappa(\alpha_1,\alpha_2)=2(\alpha_1-\alpha_2)$,\par
$h^{1,2}\left(x_{1},x_{2};\alpha_{1}\right)=\dfrac{1-\epsilon^{2} x_{2}^{2}}{2}$,\par
$h^{1,3}\left(x_{1},x_{3};\alpha_{1}-\alpha_{2}\right)=
-\dfrac{1}{2}\left(\dfrac{\left(x_{1}-x_{3}\right)^{2}}{\alpha_{1}-\alpha_{2}}-
\epsilon^{2}\left(\alpha_{1}-\alpha_{2}\right)\right)$,\par
$h^{2,4}\left(x_{2},x_{4};\alpha_{1}-\alpha_{2}\right)=
-\dfrac{1}{2}\left(\dfrac{\left(x_{2}-x_{4}\right)^{2}}{\alpha_{1}-\alpha_{2}}\right)$,\par
$\psi\left(x_{1},x_{2};\alpha_{1}\right)=
\begin{cases} \dfrac{1}{\epsilon}\log\dfrac{1+\epsilon X_{2}}{1-\epsilon X_{2}}, &
\text{if } \epsilon\neq 0, \\ 2X_{2}, & \text{if } \epsilon\neq 0, \end{cases}$ \par
$\bar{\psi}\left(x_{2},x_{1};\alpha_{1}\right)=\begin{cases}
2\dfrac{X_{1}-\epsilon^{2}\alpha_{1}X_{2}}{1-\epsilon^{2}X_{2}^{2}}, & \text{if } \epsilon\neq 0, \\
2X_{1}, & \text{if } \epsilon=0,\end{cases}$ \par
$\phi\left(x_{1},x_{3};\alpha_{1}-\alpha_{2}\right)=\begin{cases}
\dfrac{1}{\epsilon}\log\dfrac{X_{1}-X_{3}+
\epsilon\left(\alpha_{1}-\alpha_{2}\right)}{X_{1}-X_{3}-\epsilon\left(\alpha_{1}-\alpha_{2}\right)},
& \text{if } \epsilon\neq 0,\\
2\dfrac{\alpha_{1}-\alpha_{2}}{X_{1}-X_{3}}, & \text{if } \epsilon=0,\end{cases}$\par
$\bar{\phi}\left(x_{2},x_{4};\alpha_{1}-\alpha_{2}\right)=
2\dfrac{\alpha_{1}-\alpha_{2}}{X_{2}-X_{4}};$

\item[$H_{2}^{\epsilon}$:] $Q=\left(x_{1}-x_{3}\right)\left(x_{2}-x_{4}\right)+
    \left(\alpha_{2}-\alpha_{1}\right)\left(x_{1}+x_{2}+x_{3}+x_{4}\right)
    +\alpha_{2}^{2}-\alpha_{1}^{2}$\par
$\qquad-\dfrac{\epsilon}{2}\left(\alpha_{2}-\alpha_{1}\right)
\left(2x_{2}+\alpha_{1}+\alpha_{2}\right)\left(2x_{4}+\alpha_{1}+\alpha_{2}\right)-
\dfrac{\epsilon}{2}\left(\alpha_{2}-\alpha_{1}\right)^{3}$,\par
$x_{i}=\begin{cases}X_i(1+\epsilon X_{i}), &\text{if }\epsilon\neq 0\\
X_{i}&\text{if }\epsilon=0\end{cases}\quad (i=1,3),$\qquad $x_{i}=X_{i}\quad (i=2,4)$, \par
$\kappa(\alpha_1,\alpha_2)=2(\alpha_1-\alpha_2)$,\par
$h^{1,2}\left(x_{1},x_{2};\alpha_{1}\right)=
x_{1}+x_{2}+\alpha_{1}-\epsilon\left(x_{2}+\alpha_{1}\right)^{2}$,\par
$h^{1,3}\left(x_{1},x_{3};\alpha_{1}-\alpha_{2}\right)=
-\dfrac{1}{2\left(\alpha_{1}-\alpha_{2}\right)}$\par
$\qquad\cdot\left(\left(x_{1}-x_{3}\right)^{2}-\left(\alpha_{1}-
\alpha_{2}\right)^{2}-2\epsilon\left(\alpha_{1}-\alpha_{2}\right)^{2}
\left(x_{1}+x_{3}\right)+\epsilon^{2}\left(\alpha_{1}-\alpha_{2}\right)^{4}\right)$,\par
$h^{2,4}\left(x_{2},x_{4};\alpha_{1}-\alpha_{2}\right)=
-\dfrac{1}{2\left(\alpha_{1}-\alpha_{2}\right)}
\left(\left(x_{2}-x_{4}\right)^{2}-\left(\alpha_{1}-\alpha_{2}\right)^{2}\right)$,\par
$\psi\left(x_{1},x_{2};\alpha_{1}\right)=\begin{cases}
\log\dfrac{X_{1}+X_{2}+\alpha_{1}}
{1+\epsilon(X_{1}-X_{2}-\alpha_{1})}, & \text{if } \epsilon\neq 0,\\
\log\left(X_{1}+X_{2}+\alpha_{1}\right), & \text{if } \epsilon=0,
\end{cases}$,\par
$\bar{\psi}\left(x_{2},x_{1};\alpha_{1}\right)=\begin{cases}
\log(\left(X_{1}+X_{2}+\alpha_{1}\right)
\left(1+\epsilon(X_{1}-X_{2}-\alpha_{1})\right)), & \text{if } \epsilon\neq 0,\\
\log\left(X_{1}+X_{2}+\alpha_{1}\right), & \text{if } \epsilon=0,
\end{cases}$\par
$\phi\left(x_{1},x_{3};\alpha_{1}-\alpha_{2}\right)=\begin{cases}
\log\dfrac{\left(X_{1}-X_{3}+\alpha_{1}-\alpha_{2}\right)
\left(1+\epsilon(X_{1}+X_{3}+\alpha_{1}-\alpha_{2})\right)}
{\left(X_{1}-X_{3}-\alpha_{1}+\alpha_{2}\right)
\left(1+\epsilon(X_{1}+X_{3}-\alpha_{1}+\alpha_{2})\right)}, & \text{if }\epsilon\neq 0,\\
\log\dfrac{X_{1}-X_{3}+\alpha_{1}-\alpha_{2}}{X_{1}-X_{3}-\alpha_{1}+\alpha_{2}}, & \text{if } \epsilon=0,\end{cases}$\par
$\bar{\phi}\left(x_{2},x_{4};\alpha_{1}-\alpha_{2}\right)=
\log\dfrac{X_{2}-X_{4}+\alpha_{1}-\alpha_{2}}{X_{2}-X_{4}-\alpha_{1}+\alpha_{2}};$\par

\item[$H_{3}^{\epsilon}$:] $Q=e^{2\alpha_{1}}\left(x_{1}x_{2}+x_{3}x_{4}\right)-
    e^{2\alpha_{2}}\left(x_{1}x_{4}+x_{2}x_{3}\right)+
    \left(e^{4\alpha_{1}}-e^{4\alpha_{2}}\right)\left(\delta^{2}+\dfrac{\epsilon^{2} x_{2}x_{4}}{e^{2\left(\alpha_{1}+\alpha_{2}\right)}}\right)$,\par
$x_{i}=\begin{cases}2\delta\epsilon\cosh\left(2X_{i}\right), & \text{if }\epsilon\neq 0\\
\exp(2X_{i}), &\text{if }\epsilon=0 \end{cases}\quad (i=1,3),$ $\qquad x_{i}=\exp(2X_{i}),\quad (i=2,4)$,\par
$\kappa(\alpha_1,\alpha_2)=\left(e^{4\alpha_1}-e^{4\alpha_2}\right)/2$,\par
$h^{1,2}\left(x_{1},x_{2};\alpha_{1}\right)=
-2\left(x_{1}x_{2}+\delta^{2}e^{2\alpha_{1}}+\dfrac{\epsilon^{2}x_{2}^{2}}{e^{2\alpha_{1}}}\right)$,\par
$h^{1,3}\left(x_{1},x_{3};\alpha_{1}-\alpha_{2}\right)$\par
$\qquad=-\dfrac{\left(e^{\alpha_1-\alpha_2}x_{1}-
e^{\alpha_2-\alpha_1}x_{3}\right)\left(e^{\alpha_2-\alpha_1}x_{1}-
e^{\alpha_1-\alpha_2}x_{3}\right)}{\sinh(2(\alpha_1-\alpha_2))}-
4\delta^{2}\epsilon^{2}\sinh(2(\alpha_1-\alpha_2))$,\par
$h^{2,4}\left(x_{2},x_{4};\alpha_{1}-\alpha_{2}\right)=
-\dfrac{\left(e^{\alpha_1-\alpha_2}x_{2}-e^{\alpha_2-\alpha_1}x_{4}\right)
\left(e^{\alpha_2-\alpha_1}x_{2}-e^{\alpha_1-\alpha_2}x_{4}\right)}{\sinh(2(\alpha_1-\alpha_2))}$,\par
$\psi\left(x_{1},x_{2};\alpha_{1}\right)=\begin{cases}
\log\dfrac{\epsilon e^{2X_2-2X_{1}}+\delta e^{2\alpha_{1}}}
{\epsilon e^{2X_1+2X_2}+\delta e^{2\alpha_{1}}}, & \text{if } \epsilon\neq 0,\\ \\
\log\dfrac{e^{2\alpha_{1}}}{e^{2X_{1}+2X_{2}}+\delta^{2}e^{2\alpha_{1}}}, &
\text{if } \epsilon=0,\end{cases}$\par
$\bar{\psi}\left(x_{2},x_{1};\alpha_{1}\right)=\begin{cases}
\log\dfrac{e^{4\alpha_{1}}}{\left(\epsilon e^{2X_2-2X_1}+\delta e^{2\alpha_{1}}\right)\left(\epsilon e^{2X_1+2X_2}+\delta e^{2\alpha_{1}}\right)}, &
\text{if } \epsilon\neq 0,\\ \\
\log\dfrac{e^{2\alpha_{1}}}{e^{2X_1+2X_2}+\delta^{2}e^{2\alpha_{1}}}, &
\text{if } \epsilon=0,\end{cases}$\par
$\phi\left(x_{1},x_{3};\alpha_{1}-\alpha_{2}\right)=\begin{cases}
\log\dfrac{\sinh\left(X_{1}+X_{3}+\alpha_{1}-\alpha_{2}\right)
\sinh\left(X_{1}-X_{3}+\alpha_{1}-\alpha_{2}\right)}
{\sinh\left(X_{1}+X_{3}-\alpha_{1}+\alpha_{2}\right)
\sinh\left(X_{1}-X_{3}-\alpha_{1}+\alpha_{2}\right)}, &
\text{if } \epsilon\neq 0,\\ \\
\log\dfrac{\sinh(X_1-X_3+\alpha_1-\alpha_2)}{\sinh(X_1-X_3-\alpha_1+\alpha_2)}+2(\alpha_1-\alpha_2), & \text{if } \epsilon=0,\end{cases}$\par
$\bar{\phi}\left(x_{2},x_{4};\alpha_{1}-\alpha_{2}\right)=
\log\dfrac{\sinh(X_2-X_4+\alpha_1-\alpha_2)}{\sinh(X_2-X_4-\alpha_1+\alpha_2)}
+2(\alpha_1-\alpha_2)$.
\end{itemize}
\begin{rem}
The transition from $Q_3^\epsilon$ to $Q_3^0$ and from $H_3^\epsilon$ to $H_3^0$ is performed by first shifting $X_i\to X_i-\frac{1}{2}\log(\delta\epsilon)$ (for $i=1,2,3,4$ for equation $Q_3^\epsilon$ and for $i=1,3$ for equation $H_3^\epsilon$), and then sending $\epsilon\to 0$.
\end{rem}


\begin{thebibliography}{ABS09}

\bibitem[ABS03]{ABS1}
Vsevolod~E. Adler, Alexander~I. Bobenko, and Yuri~B. Suris,
  \emph{{Classification of Integrable Equations on Quad-Graphs. The Consistency
  Approach}}, Comm. Math. Phys. \textbf{233} (2003), pp. 513--543.

\bibitem[ABS09]{ABS2}
\bysame, \emph{{Discrete nonlinear hyperbolic equations. Classification of
  integrable cases}}, Funct. Anal. Appl. \textbf{43} (2009), pp. 3--17.

\bibitem[Atk08]{Atk1}
James Atkinson, \emph{B{\"a}cklund transformations for integrable lattice
  equations}, J. Phys. A: Math. Theor. \textbf{41} (2008), no.~135202.

\bibitem[BS02]{quadgraphs}
Alexander~I. Bobenko and Yuri~B. Suris, \emph{Integrable systems on
  quad-graphs}, Intern. Math. Research Notices \textbf{11} (2002), pp.
  573--611.

\bibitem[BS08]{DDG}
\bysame, \emph{{Discrete differential geometry. Integrable Struture}}, Graduate
  Studies in Mathematics, vol.~98, AMS, 2008.

\bibitem[BS10]{BS1}
\bysame, \emph{{On the Lagrangian structure of integrable quad-equations}},
  Lett. Math. Phys. \textbf{92} (2010), no.~1, pp. 17--31.
  
\bibitem[Bol11]{classification}
Raphael Boll, \emph{{Classification of 3D consistent quad-equations}}, to
  appear in J. Nonlin. Math. Phys. (2011).  

\bibitem[BolS10]{Todapaper}
Raphael Boll and Yuri~B. Suris, \emph{{Non-symmetric discrete Toda systems from
  quad-graphs}}, Applicable Analysis \textbf{89} (2010), no.~4, pp. 547--569.

\bibitem[Hie04]{Hietarinta}
Jarmo Hietarinta, \emph{A new two-dimensional lattice model that is 'consistent
  around a cube'}, J. Phys. A: Math. Theor. \textbf{37} (2004), pp. 67--73.

\bibitem[HV10]{HV}
Peter~E. Hydon and Claude-M. Viallet, \emph{Asymmetric integrable quad-graph
  equations}, Ap\-pli\-ca\-ble Analysis \textbf{89} (2010), no.~4, pp.
  493--506.

\bibitem[LN09]{LN}
Sarah Lobb and Frank~W. Nijhoff, \emph{Lagrangian multiforms and
  multidimensional consistency}, J. Phys. A: Math. Theor. \textbf{42} (2009),
  no.~454013.

\bibitem[LY09]{LY}
Decio Levi and Ravil~I. Yamilov, \emph{On a linear inegrable difference
  equation on the square}, Ufa Math. J. \textbf{1} (2009), no.~2, pp. 101--105.

\bibitem[MV91]{MV}
J{\"u}rgen Moser and Aleksandr~P. Veselov, \emph{Discrete versions of some
  classical integrable systems and factorization of matrix polynomials}, Comm.
  Math. Phys. \textbf{139} (1991), pp. 217--243.

\bibitem[Nij02]{Nijhoff}
Frank~W. Nijhoff, \emph{{Lax pair for Adler (lattice Krichever-Novikov)
  system}}, Phys. Lett. A \textbf{297} (2002), pp. 49--58.

\bibitem[NW01]{NW}
Frank~W. Nijhoff and Alan~J. Walker, \emph{{The discrete and continous
  Painlev{\'e} VI hierachy and the Garnier systems}}, Glasg. Math. J.
  \textbf{43A} (2001), pp. 109--123.

\bibitem[XP09]{XP}
Pavlos~D. Xenitidis and Vassilis~G. Papageorgiou, \emph{Symmetries and
  integrability of discrete equations defined on a black-white lattice}, J.
  Phys. A: Math. Theor. \textbf{42} (2009), no.~454025.

\end{thebibliography}

\providecommand{\bysame}{\leavevmode\hbox to3em{\hrulefill}\thinspace}
\providecommand{\MR}{\relax\ifhmode\unskip\space\fi MR }
\providecommand{\MRhref}[2]{%
  \href{http://www.ams.org/mathscinet-getitem?mr=#1}{#2}
}
\providecommand{\href}[2]{#2}

\end{document}